\documentclass[runningheads,a4paper]{llncs}

\usepackage{amssymb}
\usepackage{graphicx}
\usepackage{url}
\usepackage{amsmath}
\usepackage{tikz}
\usetikzlibrary{calc}
\usetikzlibrary{arrows.meta,positioning}

\begin{document}

\mainmatter 


\title{Even more properties of parity based bit-counting complexity classes}

\titlerunning{More properties of Parity based bit-counting complexity classes }

\author{ Tayfun Pay}

\authorrunning{Tayfun Pay}

\urldef{\mailsa}\path|tpay@gradcenter.cuny.edu|

\institute{{ \mailsa}}



%
%

\maketitle

\begin{abstract}

We investigate several additional properties of parity based bit-counting complexity classes ${\bf B_{|0|\oplus}P}$ and ${\bf B_{|1|\oplus}P}$. We first prove that ${\bf C_=P}={\bf ES}={\bf MNS}\subseteq{\bf P}^{{\bf B_{|1|\oplus}P}}={\bf P}^{{\bf B_{|0|\oplus}P}}$. We then build upon the four consecutive value theorem from \cite{P26b} and develop a bit recovery procedure that reconstructs the binary representation of a ${\bf \#P}$ value one bit at a time by querying appropriately padded and shifted versions of that value. Using this procedure, we prove that ${\bf PP}\subseteq{\bf P}^{{\bf B_{|1|\oplus}P}}$ and ${\bf PP}\subseteq{\bf P}^{{\bf B_{|0|\oplus}P}}$, which consequently yields ${\bf P}^{\bf PP}={\bf P}^{\bf B_{|0|\oplus}P}={\bf P}^{\bf B_{|1|\oplus}P}$. Furthermore, since this procedure recovers the complete numerical output of a ${\bf \#P}$ function, it also establishes the functional containments of ${\bf \#P}\subseteq{\bf FP}^{{\bf B_{|1|\oplus}P}}$ and ${\bf \#P}\subseteq{\bf FP}^{{\bf B_{|0|\oplus}P}}$. We additionally demonstrate that caching repeated oracle queries reduces the number of distinct queries required for the bit recovery procedure from $4n$ to $2n+2$. Finally, we prove that parity based bit-counting hierarchies ${\bf \Sigma B_{|0|\oplus}H}$, ${\bf \Sigma B_{|1|\oplus}H}$, ${\bf \Delta B_{|0|\oplus}H}$ and ${\bf \Delta B_{|1|\oplus}H}$ contain ${\bf CH}$. As a result, these hierarchies coincide with ${\bf CH}$ since the reverse containments were already established in \cite{P26b}. 

\end{abstract}

\section{Introduction}

We study several additional properties of parity based bit-counting complexity classes ${\bf B_{|0|\oplus}P}$  and ${\bf B_{|1|\oplus}P}$ , which were introduced in \cite{P26a} and further investigated in \cite{P26b}. These complexity classes differ from classical counting complexity classes in their acceptance criteria. For example, ${\bf PP}$ determines membership by comparing the number of accepting computation paths with the number of rejecting computation paths, whereas ${\bf B_{|0|\oplus}P}$  and ${\bf B_{|1|\oplus}P}$  determine membership by examining the parity of the number of 0's bits and 1's bits, respectively, in the binary representation of a ${\bf \# P}$ value. Although these acceptance criteria examine different numerical properties, we show that polynomial time access to either parity based bit-counting complexity class is sufficient to recover the complete value of a ${\bf \# P}$ function. This leads to several containments as well as equalities involving classical counting complexity classes and their associated hierarchies.

We first investigate the relationship between the parity of the binary length of an integer and the parity of the number of 0's bits and 1's bits in its binary representation. We show that the exclusive-or of these two bit-counting parities equals the parity of the binary length. Furthermore, this value changes between an integer and its successor precisely when the integer is a Mersenne number. We use this fact to prove that ${\bf MNS}\subseteq{\bf P}^{\bf B_{|0|\oplus}P}={\bf P}^{\bf B_{|1|\oplus}P}$. This consequently yields ${\bf ES}\subseteq{\bf P}^{\bf B_{|0|\oplus}P}={\bf P}^{\bf B_{|1|\oplus}P}$ and ${\bf C_=P}\subseteq{\bf P}^{\bf B_{|0|\oplus}P}={\bf P}^{\bf B_{|1|\oplus}P}$ since it is already known that ${\bf MNS}={\bf ES}={\bf C_=P}$.

We then build upon the four consecutive value theorems from \cite{P26b}, which were previously used to show that ${\bf \oplus P}\subseteq {\bf P}^{{\bf B_{|0|\oplus}P}}$ and ${\bf \oplus P}\subseteq {\bf P}^{{\bf B_{|1|\oplus}P}}$. These theorems demonstrate that the parity of an integer can be determined from four consecutive values of either the Prouhet-Thue-Morse sequence or the zero bit parity sequence. We extend this approach by introducing a bit recovery procedure that uses padding and shifted versions of a ${\bf \# P}$ value to recover its binary digits one at a time. In each round, the already recovered lower order bits allow the deterministic machine to determine the contribution of those bits to the parity of the queried values. After removing this known contribution, the four consecutive value theorem determines the next binary digit. Repeating this procedure recovers the complete value of the counting function.

This construction yields two important consequences. First, since the number of accepting paths of a non-deterministic polynomial time Turing machine can be recovered, a deterministic polynomial time oracle machine can determine whether that number exceeds the majority threshold. This establishes the containments ${\bf PP}\subseteq{\bf P}^{\bf B_{|0|\oplus}P}$ and ${\bf PP}\subseteq{\bf P}^{\bf B_{|1|\oplus}P}$. Second, the same procedure can be implemented by a deterministic polynomial time oracle transducer to compute the complete output of any ${\bf \#P}$ function. This establishes the functional containments ${\bf \#P}\subseteq{\bf FP}^{\bf B_{|0|\oplus}P}$ and ${\bf \#P}\subseteq{\bf FP}^{\bf B_{|1|\oplus}P}$. Although the straightforward procedure uses $4n$ oracle queries to recover $n$ binary digits, we also demonstrate that repeated numerical queries can be cached, reducing the number of distinct queried values to $2n+2$.

We then show that ${\bf P}^{\bf PP}= {\bf P}^{{\bf B_{|0|\oplus}P}} = {\bf P}^{{\bf B_{|1|\oplus}P}}$. When this is combined with the results from \cite{P26a}, namely ${\bf P}^{\bf PP}={\bf P}^{{\bf B_{|0|=|1|}P}}={\bf P}^{{\bf B_{|0|<|1|}P}}={\bf P}^{{\bf B_{|0|>|1|}P}}$, we establish that parity based bit-counting complexity classes and comparison based bit-counting complexity classes are all equivalent to ${\bf PP}$ under polynomial time Turing reductions. This provides another illustration of the observation in \cite{KPZ99} that polynomial time Turing reductions blur structural differences between complexity classes, even when their underlying acceptance criteria appear substantially different.

Finally, we investigate the hierarchies generated by parity based bit-counting complexity classes. We show that the bit recovery construction relativizes, allowing computations in counting hierarchy to be simulated within the corresponding parity based bit-counting hierarchies. In particular, we prove that  ${\bf CH}\subseteq{\bf \Sigma B_{|0|\oplus}H} = {\bf \Sigma B_{|1|\oplus}H}$ and ${\bf CH}\subseteq{\bf \Delta B_{|0|\oplus}H} = {\bf \Delta B_{|1|\oplus}H}$. Since the reverse containments were previously established in \cite{P26b}, these results show that both parity based bit-counting hierarchies coincide with counting hierarchy when their levels are united. Thus, parity based bit-counting complexity classes not only provide the same computational power as ${\bf PP}$ under polynomial time Turing reductions, but also generate hierarchies that characterize ${\bf CH}$.

\section{Definitions and containments}

\subsection{Some classical complexity classes and hierarchies}

\begin{definition}\normalfont
A language $L$ is in complexity class {\bf NP}, if there exists a polynomial $p$ and a polynomial time predicate $R$ such that, for each $x$,

$x \in L \Leftrightarrow ||\{y| \ |y| = p(|x|) \wedge R(x, y)\}|| >0 $
\end{definition}

\begin{definition}\normalfont
A language $L$ is in complexity class {\bf CoNP}, if there exists a polynomial $p$ and a polynomial time predicate $R$ such that, for each $x$,

$x \in L \Leftrightarrow ||\{y| \ |y| = p(|x|) \wedge R(x, y)\}|| =0 $
\end{definition}

\begin{definition}\normalfont
Polynomial hierarchy, denoted ${\bf PH}$ and as defined in \cite{S76}, is the union of the levels ${\bf \Sigma_k^P}$, ${\bf \Pi_k^P}$, and ${\bf \Delta_k^P}$. The zeroth levels are ${\bf \Sigma_0^P}={\bf \Pi_0^P}={\bf \Delta_0^P}={\bf P}$. For every $k\geq 0$, define ${\bf \Sigma_{k+1}^P}={\bf NP}^{{\bf \Sigma_k^P}}$, ${\bf \Pi_{k+1}^P}={\bf coNP}^{{\bf \Sigma_k^P}}$, and ${\bf \Delta_{k+1}^P}={\bf P}^{{\bf \Sigma_k^P}}$. The full polynomial hierarchy is ${\bf PH}=\bigcup_{k\geq 0}{\bf \Sigma_k^P}=\bigcup_{k\geq 0}{\bf \Pi_k^P}=\bigcup_{k\geq 0}{\bf \Delta_k^P}$.
\end{definition}

This first few levels of ${\bf PH}$ are depicted in the figure below.

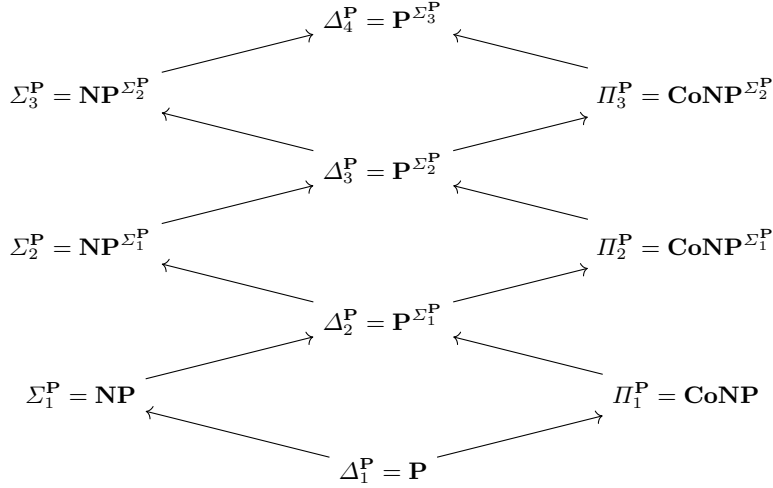
\begin{tikzpicture}[
    every node/.style={font=\small},
    classnode/.style={align=center}
]

\node[classnode] (d1) at (0,0) {$\Delta_1^{\bf P}={\bf P}$};

\node[classnode] (s1) at (-4,1) {$\Sigma_1^{\bf P}={\bf NP}$};
\node[classnode] (p1) at (4,1) {$\Pi_1^{\bf P}={\bf CoNP}$};

\node[classnode] (d2) at (0,2) {$\Delta_2^{\bf P}={\bf P}^{\Sigma_1^{\bf P}}$};

\node[classnode] (s2) at (-4,3) {$\Sigma_2^{\bf P}={\bf NP}^{\Sigma_1^{\bf P}}$};
\node[classnode] (p2) at (4,3) {$\Pi_2^{\bf P}={\bf CoNP}^{\Sigma_1^{\bf P}}$};

\node[classnode] (d3) at (0,4) {$\Delta_3^{\bf P}={\bf P}^{\Sigma_2^{\bf P}}$};

\node[classnode] (s3) at (-4,5) {$\Sigma_3^{\bf P}={\bf NP}^{\Sigma_2^{\bf P}}$};
\node[classnode] (p3) at (4,5) {$\Pi_3^{\bf P}={\bf CoNP}^{\Sigma_2^{\bf P}}$};

\node[classnode] (d4) at (0,6) {$\Delta_4^{\bf P}={\bf P}^{\Sigma_3^{\bf P}}$};

\draw[->] (d1) -- (s1);
\draw[->] (d1) -- (p1);

\draw[->] (s1) -- (d2);
\draw[->] (p1) -- (d2);

\draw[->] (d2) -- (s2);
\draw[->] (d2) -- (p2);

\draw[->] (s2) -- (d3);
\draw[->] (p2) -- (d3);

\draw[->] (d3) -- (s3);
\draw[->] (d3) -- (p3);

\draw[->] (s3) -- (d4);
\draw[->] (p3) -- (d4);

\end{tikzpicture}

\begin{definition}\normalfont
A language $L$ is in complexity class {\bf $\oplus$P}, as defined in \cite{PZ83}, if there exists a polynomial $p$ and a polynomial time predicate $R$ such that, for each $x$,

$x \in L \Leftrightarrow ||\{y| \ |y| = p(|x|) \wedge R(x, y)\}|| \not \equiv$ {\rm 0 (Mod 2)}
\end{definition}

\begin{definition}\normalfont
A language $L$ is in complexity class {\rm \bf C$_{=}$P}, as defined in \cite{S75}, if there exists a polynomial $p$ and a polynomial time predicate $R$ such that, for each $x$,
$x \in L \Leftrightarrow ||\{y| \ |y| = p(|x|) \wedge R(x, y)\}|| = 2^{p(|x|)-1} $
\end{definition}

\begin{definition}\normalfont
A language $L$ is in complexity class {\rm \bf ES}, as defined in \cite{BHR00}, if there exists a polynomial $p$ and a polynomial time predicate $R$ such that, for each $x$,

$x \in L \Leftrightarrow ||\{y| \ |y| = p(|x|) \wedge R(x, y)\}|| = 2^{t}$, where t $\in \mathbb{N}_{0} = \{0,1,2, ...\}$
\end{definition}

\begin{definition}\normalfont A language $L$ is in complexity class {\bf MNS}, as defined in \cite{CP18}, if there exists a polynomial $p$ and a polynomial time predicate $R$ such that, for each $x$,

$x \in L \Leftrightarrow ||\{y| \ |y| = p(|x|) \wedge R(x, y)\}|| = 2^{t}-1$, where t $\in \mathbb{N}_{>0} = \{1,2,3, ...\}$
\end{definition}

\begin{definition}\normalfont
A language $L$ is in complexity class {\rm \bf PP}, as defined in \cite{S75}, if there exists a polynomial $p$ and a polynomial time predicate $R$ such that, for each $x$,

$x \in L \Leftrightarrow ||\{y| \ |y| = p(|x|) \wedge R(x, y)\}|| > 2^{p(|x|)-1} $

\end{definition}

\begin{definition}\normalfont

Counting hierarchy, denoted ${\bf CH}$, was introduced in \cite{W86}. We use the oracle characterization of counting hierarchy defined in \cite{T88}. Counting hierarchy is the union of the levels ${\bf C_kP}$. The zeroth level is ${\bf C_0P}={\bf P}$. For every $k\geq 0$, the next level is ${\bf C_{k+1}P}={\bf PP}^{{\bf C_kP}}$, where ${\bf PP}^{{\bf C_kP}}$ denotes the class of languages decidable by a ${\bf PP}$ machine with oracle access to a language in ${\bf C_kP}$. The full counting hierarchy is ${\bf CH}=\bigcup_{k\geq 0}{\bf C_kP}$. 
\end{definition}

The first few levels of the ${\bf CH}$ are ${\bf C_0P}={\bf P}$, ${\bf C_1P}={\bf PP}$, ${\bf C_2P}={\bf PP}^{\bf PP}$, and ${\bf C_3P}={\bf PP}^{{\bf PP}^{\bf PP}}$.

\begin{definition}\normalfont
Functional complexity class {\rm \bf \#P}, as defined in \cite{V79}, counts the total number of accepting paths of a non-deterministic polynomial time Turing machine.

{\bf \#P} = $\{f | (\exists$ a non-deterministic polynomial time Turing machine $M) (\forall x)$ $[f (x) = \#accept_{M}(x)]\}$.
\end{definition}

It is known that ${\bf CoNP} \subseteq {\bf MNS} = {\bf ES } = { \bf C_{=}P} \subseteq {\bf PP}$ and that ${\bf PH}\subseteq{\bf P}^{\bf PP} = {\bf P}^{{\bf\#P}} \subseteq {\bf CH}$, where it was shown in \cite{T89} that ${\bf PH}\subseteq{\bf P}^{\bf PP}$. Various other complexity classes can be found in \cite{HO02} and \cite{CP18}.

\subsection{Bit-counting complexity classes}
The function $B_{0}$ counts the number of 0's bits and the function $B_{1}$ counts the number of 1's bits in the binary representation of the numbers in the following definitions. Note that $B_{0}(0)=1$. 

\begin{definition}\normalfont
A language $L$ is in complexity class ${\bf B_{|0|=|1|}P}$, if there exist a polynomial $p$ and a polynomial time predicate $R$ such that, for each $x$, 

$x \in L \Leftrightarrow B_{0}(||\{y| \ |y| = p(|x|) \wedge R(x, y)\}||)$ $=$ $B_{1}(||\{y| \ |y| = p(|x|) \wedge R(x, y)\}||)$
\end{definition}

\begin{definition}\normalfont
A language $L$ is in complexity class ${\bf B_{|0|>|1|}P}$, if there exist a polynomial $p$ and a polynomial time predicate $R$ such that, for each $x$, 

$x \in L \Leftrightarrow B_{0}(||\{y| \ |y| = p(|x|) \wedge R(x, y)\}||)$ $>$ $B_{1}(||\{y| \ |y| = p(|x|) \wedge R(x, y)\}||)$
\end{definition}

\begin{definition}\normalfont
A language $L$ is in complexity class ${\bf B_{|0|<|1|}P}$, if there exist a polynomial $p$ and a polynomial time predicate $R$ such that, for each $x$, 

$x \in L \Leftrightarrow B_{0}(||\{y| \ |y| = p(|x|) \wedge R(x, y)\}||)$ $<$ $B_{1}(||\{y| \ |y| = p(|x|) \wedge R(x, y)\}||)$
\end{definition}

\begin{definition}\normalfont
A language $L$ is in complexity class ${\bf B_{|0| \oplus}P}$, if there exist a polynomial $p$ and a polynomial time predicate $R$ such that, for each $x$,  

$x \in L \Leftrightarrow B_{0}(||\{y| \ |y| = p(|x|) \wedge R(x, y)\}||) \not \equiv  {\rm 0 (Mod 2)}$
\end{definition}

\begin{definition}\normalfont
A language $L$ is in complexity class ${\bf B_{|1| \oplus}P}$, if there exist a polynomial $p$ and a polynomial time predicate $R$ such that, for each $x$,  

$x \in L \Leftrightarrow B_{1}(||\{y| \ |y| = p(|x|) \wedge R(x, y)\}||) \not \equiv  {\rm 0 (Mod 2)}$
\end{definition}

\begin{definition}\normalfont
A language $L$ is in complexity class ${\bf B_{|1|=0}P}$, if there exist a polynomial $p$ and a polynomial time predicate $R$ such that, for each $x$,
 
$x \in L \Leftrightarrow B_{1}(||\{y| \ |y| = p(|x|) \wedge R(x, y)\}||)$ $= 0$
\end{definition}

\begin{definition}\normalfont
A language $L$ is in complexity class ${\bf B_{|1|>0}P}$, if there exist a polynomial $p$ and a polynomial time predicate $R$ such that, for each $x$,
 
$x \in L \Leftrightarrow B_{1}(||\{y| \ |y| = p(|x|) \wedge R(x, y)\}||)$ $> 0$
\end{definition}

\begin{definition}\normalfont
A language $L$ is in complexity class ${\bf B_{|0|=0}P}$, if there exist a polynomial $p$ and a polynomial time predicate $R$ such that, for each $x$, 

$x \in L \Leftrightarrow B_{0}(||\{y| \ |y| = p(|x|) \wedge R(x, y)\}||)$ $= 0$
\end{definition}

\begin{definition}\normalfont
A language $L$ is in complexity class ${\bf B_{|0|>0}P}$, if there exist a polynomial $p$ and a polynomial time predicate $R$ such that, for each $x$, 

$x \in L \Leftrightarrow B_{0}(||\{y| \ |y| = p(|x|) \wedge R(x, y)\}||)$ $> 0$
\end{definition}

\begin{definition}\normalfont
The $\Sigma$-${\bf B_{|0|\oplus}P}$ hierarchy, denoted ${\bf \Sigma B_{|0|\oplus}H}$, is defined as follows. The zeroth level is ${\bf \Sigma B_{|0|\oplus}^{0}}={\bf P}$. The first level is ${\bf \Sigma B_{|0|\oplus}^{1}}={\bf B_{|0|\oplus}P}$. For every $k\geq 1$, define the next level by ${\bf \Sigma B_{|0|\oplus}^{k+1}}={\bf B_{|0|\oplus}P}^{{\bf \Sigma B_{|0|\oplus}^{k}}}$. The full hierarchy is ${\bf \Sigma B_{|0|\oplus}H}=\bigcup_{k\geq 0}{\bf \Sigma B_{|0|\oplus}^{k}}$.
\end{definition}

\begin{definition}\normalfont
The $\Delta$-${\bf B_{|0|\oplus}P}$ hierarchy, denoted ${\bf \Delta B_{|0|\oplus}H}$, is defined as follows. The zeroth level is ${\bf \Delta B_{|0|\oplus}^{0}}={\bf P}$. For every $k\geq 1$, define ${\bf \Delta B_{|0|\oplus}^{k}}={\bf P}^{{\bf \Sigma B_{|0|\oplus}^{k}}}$. The full hierarchy is ${\bf \Delta B_{|0|\oplus}H}=\bigcup_{k\geq 0}{\bf \Delta B_{|0|\oplus}^{k}}$.
\end{definition}

\begin{definition}\normalfont
The $\Sigma$-${\bf B_{|1|\oplus}P}$ hierarchy, denoted ${\bf \Sigma B_{|1|\oplus}H}$, is defined as follows. The zeroth level is ${\bf \Sigma B_{|1|\oplus}^{0}}={\bf P}$. The first level is ${\bf \Sigma B_{|1|\oplus}^{1}}={\bf B_{|1|\oplus}P}$. For every $k\geq 1$, define the next level by ${\bf \Sigma B_{|1|\oplus}^{k+1}}={\bf B_{|1|\oplus}P}^{{\bf \Sigma B_{|1|\oplus}^{k}}}$. The full hierarchy is ${\bf \Sigma B_{|1|\oplus}H}=\bigcup_{k\geq 0}{\bf \Sigma B_{|1|\oplus}^{k}}$.
\end{definition}

\begin{definition}\normalfont
The $\Delta$-${\bf B_{|1|\oplus}P}$ hierarchy, denoted ${\bf \Delta B_{|1|\oplus}H}$, is defined as follows. The zeroth level is ${\bf \Delta B_{|1|\oplus}^{0}}={\bf P}$. For every $k\geq 1$, define ${\bf \Delta B_{|1|\oplus}^{k}}={\bf P}^{{\bf \Sigma B_{|1|\oplus}^{k}}}$. The full hierarchy is ${\bf \Delta B_{|1|\oplus}H}=\bigcup_{k\geq 0}{\bf \Delta B_{|1|\oplus}^{k}}$.
\end{definition}

\subsection{Definitions and theorems related to parity based bit-counting complexity classes}

\begin{definition}\normalfont
The Prouhet-Thue-Morse sequence \cite{M21} \cite{S96}  is the infinite binary sequence $(t(m))_{m\ge 0}$ defined by $t(m)\equiv B_1(m)\pmod 2$, where $B_1(m)$ denotes the number of $1$'s bits in the standard binary representation of $m$. Equivalently, $t(m)=0$ if $B_1(m)$ is even, and $t(m)=1$ if $B_1(m)$ is odd. The sequence starts with $(t(m))_{m\ge 0}=0,$$1,1,0,1,0,0,1,1,0,0,1,0,1,1,0,\ldots$. 

It further satisfies the recursive identities $t(0)=0$, $t(2m)=t(m)$, and $t(2m+1)=1-t(m)$. In essence, the sequence contains no three consecutive equal bits, neither $000$ nor $111$ occurs as a contiguous block.

\end{definition}

It was proven in \cite{P26b} that for every $N\geq 0$, the four values $t(N),t(N+1),t(N+2),t(N+3)$ determine whether $N$ is even or odd. More precisely, $N$ is even if and only if $t(N)\neq t(N+1)$ and $t(N+2)\neq t(N+3)$.

\begin{definition}\normalfont
The $B_0$ parity sequence \cite{B01} is the infinite binary sequence $(s(m))_{m\ge 0}$ defined by $s(m)\equiv B_0(m)\pmod 2$, where $B_0(m)$ denotes the number of $0$'s bits in the standard binary representation of $m$, with the convention that the standard binary representation of $0$ is $0$. Equivalently, $s(m)=0$ if $B_0(m)$ is even, and $s(m)=1$ if $B_0(m)$ is odd. The sequence starts with $(s(m))_{m\ge 0}=1,$$0,1,0,0,1,1,0,1,0,0,1,0,1,1,0,\ldots$.

It further satisfies the recursive identities $s(0)=1$, $s(1)=0$, $s(2m)=1-s(m)$ for $m\ge 1$, and $s(2m+1)=s(m)$ for $m\ge 1$. In essence, the sequence contains no three consecutive equal bits, neither $000$ nor $111$ occurs as a contiguous block.
\end{definition}

It was proven in \cite{P26b} that for every $N\geq 0$, the four values $s(N),s(N+1),s(N+2),s(N+3)$ determine whether $N$ is even or odd. More precisely, $N$ is even if and only if $s(N)\neq s(N+1)$ and $s(N+2)\neq s(N+3)$.

\subsection{Some known containments and equalities }

The following equalities follow from their respective definitions.  

-${\bf B_{|1|=0}P} ={\bf CoNP}$.

-${\bf B_{|1|>0}P} ={\bf NP}$.

-${\bf B_{|0|=0}P} = {\bf MNS}$.
 
-${\bf B_{|0|>0}P} = {\bf CoMNS}$.  
\newline \newline
The following containments and equalities were proven in \cite{P26a}. 

-${\bf NP}\subseteq {\bf B_{|1| \oplus}P}$

-${\bf NP}\subseteq {\bf B_{|0| \oplus}P}$

-${\bf CoNP}\subseteq {\bf B_{|1| \oplus}P}$

-${\bf CoNP}\subseteq {\bf B_{|0| \oplus}P}$

-${\bf B_{|1| \oplus}P}\subseteq {\bf P}^{\bf PP}$

-${\bf B_{|0| \oplus}P}\subseteq {\bf P}^{\bf PP}$

-${\bf B_{|1|\oplus}P}\subseteq {\bf P}^{{\bf B_{|0|\oplus}P}}$

-${\bf B_{|0|\oplus}P}\subseteq {\bf P}^{{\bf B_{|1|\oplus}P}}$

-${\bf P}^{{\bf B_{|0|\oplus}P}}={\bf P}^{{\bf B_{|1|\oplus}P}}$

-${\bf PP}\subseteq {\bf B_{|0|=|1|}P}$

-${\bf PP}\subseteq {\bf B_{|0|>|1|}P}$

-${\bf PP}\subseteq {\bf B_{|0|<|1|}P}$

-${\bf B_{|0|=|1|}P}\subseteq {\bf P}^{\bf PP}$

-${\bf B_{|0|>|1|}P}\subseteq {\bf P}^{\bf PP}$

-${\bf B_{|0|<|1|}P}\subseteq {\bf P}^{\bf PP}$

-${\bf P}^{\bf PP} = {\bf P}^{{\bf B_{|0|=|1|}P}} = {\bf P}^{{\bf B_{|0|>|1|}P}}={\bf P}^{{\bf B_{|0|<|1|}P}}$
\newline \newline
The following containments and equalities were proven in \cite{P26b}. 

-${\bf B_{|1|\oplus}P}={\bf CoB_{|1|\oplus}P}$

-${\bf B_{|0|\oplus}P}={\bf CoB_{|0|\oplus}P}$

-${\bf B_{|1|\oplus}P}\subseteq {\bf B_{|0|\oplus}P}$

-${\bf \oplus P}\subseteq {\bf P}^{{\bf B_{|1|\oplus}P}}$.

-${\bf \oplus P}\subseteq {\bf P}^{{\bf B_{|0|\oplus}P}}$.

-${\bf \Sigma B_{|0|\oplus}H}={\bf \Sigma B_{|1|\oplus}H}$

-${\bf \Delta B_{|0|\oplus}H}={\bf \Delta B_{|1|\oplus}H}$

-${\bf PH}\subseteq{\bf \Sigma B_{|0|\oplus}H} = {\bf \Sigma B_{|1|\oplus}H}$

-${\bf PH}\subseteq{\bf \Delta B_{|0|\oplus}H} = {\bf \Delta B_{|1|\oplus}H}$

-${\bf \Sigma B_{|0|\oplus}H}={\bf \Sigma B_{|1|\oplus}H}\subseteq{\bf CH}$

-${\bf \Delta B_{|0|\oplus}H}={\bf \Delta B_{|1|\oplus}H}\subseteq{\bf CH}$

\section{Several properties of parity based bit-counting complexity classes}

\subsection{Containment of ${\bf C_=P}={\bf ES}={\bf MNS}$}

We first prove that when you xor the parity of the number of $1$'s bits with the parity of the number of $0$'s bits, you get exactly the parity of the binary length. So the value $t(N)\oplus s(N)$ changes between $N$ and $N+1$ precisely when the binary lengths of $N$ and $N+1$ have different parity. We then show that this occurs exactly when $N$ is a Mersenne number. We then prove the containment of ${\bf MNS}$ in ${\bf P}^{\bf B_{|0|\oplus}P}={\bf P}^{\bf B_{|1|\oplus}P}$. 

\begin{theorem} For every $N\geq 0$, $(t(N)\oplus s(N))\neq(t(N+1)\oplus s(N+1))$ if and only if $\ell(N)\not\equiv \ell(N+1)\pmod 2$
\end{theorem}

\begin{proof}

Let $t(N)=B_1(N)\pmod 2$ and $s(N)=B_0(N)\pmod 2$, where $B_1(N)$ denotes the number of $1$'s bits and $B_0(N)$ denotes the number of $0$'s bits in the standard binary representation of $N$. Let $\ell(N)=B_1(N)+B_0(N)$ denote the length of the standard binary representation of $N$. 

First note that both $t(N)$ and $s(N)$ output either $0$ or $1$. Then for any two bits $a,b\in\{0,1\}$, the exclusive or operation satisfies $a\oplus b\equiv a+b\pmod 2$. To see this, observe the four possible cases as follows: If $a=b=0$, then $a\oplus b=0$ and $a+b=0$. If exactly one of $a,b$ is $1$, then $a\oplus b=1$ and $a+b=1$. If $a=b=1$, then $a\oplus b=0$ and $a+b=2\equiv 0\pmod 2$. Thus, for bits, exclusive-or is addition modulo $2$.

Applying this gives $t(N)\oplus s(N)\equiv t(N)+s(N)\pmod 2$. We already know $B_1(N)+B_0(N)=\ell(N)$. Substituting this gives $t(N)\oplus s(N)\equiv \ell(N)\pmod 2$. Similarly for $N+1$, which gives $t(N+1)\oplus s(N+1)\equiv \ell(N+1)\pmod 2$.

Now let $c=t(N)\oplus s(N)$ and $d=t(N+1)\oplus s(N+1)$, where $c,d\in\{0,1\}$. We claim that $c\neq d$ if and only if $c\not\equiv d\pmod 2$. To see this, observe the four possible cases as follows: If $c=d=0$, then $c=d$ and $c\equiv d\pmod 2$. If $c=d=1$, then $c=d$ and $c\equiv d\pmod 2$. If $c=0$ and $d=1$, then $c\neq d$ and $c\not\equiv d\pmod 2$. If $c=1$ and $d=0$, then $c\neq d$ and $c\not\equiv d\pmod 2$. Thus, for bits, being different is exactly the same as having different residues modulo $2$.

Therefore, for every $N\geq 0$, $(t(N)\oplus s(N))\neq(t(N+1)\oplus s(N+1))$ if and only if $t(N)\oplus s(N)\not\equiv t(N+1)\oplus s(N+1)\pmod 2$. Using $t(N)\oplus s(N)\equiv \ell(N)\pmod 2$ and $t(N+1)\oplus s(N+1)\equiv \ell(N+1)\pmod 2$, this is equivalent to $\ell(N)\not\equiv \ell(N+1)\pmod 2$.$\qed$
\end{proof}

\begin{corollary}
$N$ is a Mersenne number if and only if $\ell(N)\not\equiv \ell(N+1)\pmod 2$ for every $N\geq 0$
\end{corollary}

\begin{proof}
A Mersenne number is a number of the form $2^m-1$ for $m\geq 1$.

We first assume that $N=0$. The standard binary expansion of $0$ is $0$, so $\ell(0)=1$ and $\ell(1)=1$ then $\ell(0)\equiv \ell(1)\pmod 2$. As a result the condition that $\ell(N)\not\equiv \ell(N+1)\pmod 2$ does not hold when $N=0$. This agrees with our convention that $0$ is not a Mersenne number.

We now assume that $N\geq 1$. Let $r=\ell(N)$. Then the standard binary expansion of $N$ has length $r$, so $2^{r-1}\leq N\leq 2^r-1$. 

First assume $N<2^r-1$. Since $N$ is an integer, this means that $N\leq 2^r-2$ and $N+1\leq 2^r-1$. Also, since $N\geq 2^{r-1}$, we have $N+1\geq 2^{r-1}+1>2^{r-1}$. Thus $2^{r-1}\leq N+1\leq 2^r-1$ and $\ell(N+1)=r$. Thus $\ell(N+1)=\ell(N)$, and as a result $\ell(N)\equiv \ell(N+1)\pmod 2$.

Next assume $N=2^r-1$ and $N+1=2^r$. The standard binary expansion of $2^r$ is $1$ followed by $r$ many $0$'s, so $\ell(N+1)=r+1$. Since $\ell(N)=r$, we get $\ell(N)\not\equiv \ell(N+1)\pmod 2$.

Therefore, for every $N\geq 1$, $\ell(N)\not\equiv \ell(N+1)\pmod 2$ holds if and only if $N=2^r-1$, a Mersenne number, where $r=\ell(N)\geq 1$. Equivalently, this holds if and only if $N$ is of the form $2^m-1$ for some $m\geq 1$. Together with the $N=0$ case, this proves that, for every $N\geq 0$, $N$ is a Mersenne number if and only if $\ell(N)\not\equiv \ell(N+1)\pmod 2$.$\qed$
\end{proof}

\begin{theorem}
${\bf C_{=}P}={\bf ES}={\bf MNS}\subseteq {\bf P}^{\bf B_{|0|\oplus}P}={\bf P}^{\bf B_{|1|\oplus}P}$
\end{theorem}

\begin{proof}
We use the established equality ${\bf C_{=}P}={\bf ES}={\bf MNS}$. Therefore, it is enough to prove that ${\bf MNS}\subseteq {\bf P}^{\bf B_{|0|\oplus}P}={\bf P}^{\bf B_{|1|\oplus}P}$.

Let $L\in{\bf MNS}$. Then there is a ${\bf \#P}$ function $F(x)$ such that, for every input $x$, $x\in L$ if and only if $F(x)$ is a Mersenne number. We use the convention that a Mersenne number is a number of the form $2^m-1$ for some $m\geq 1$, so $0$ is not a Mersenne number. Define $F_0(x)=F(x)$ and $F_1(x)=F(x)+1$. Both $F_0$ and $F_1$ are ${\bf \#P}$ functions since ${\bf \# P}$ is closed under addition of constants.

The two functions $F_0$ and $F_1$ can be combined into a single oracle language $O_t\in{\bf B_{|1|\oplus}P}$ by tagging the query with $j\in\{0,1\}$. The oracle says yes on $(x,j)$ exactly when $B_1(F_j(x))\not\equiv 0\pmod 2$. Thus the oracle answer on $(x,j)$ is precisely $t(F(x)+j)$. Similarly, the two functions $F_0$ and $F_1$ can be combined into a single oracle language $O_s\in{\bf B_{|0|\oplus}P}$ by tagging the query with $j\in\{0,1\}$. The oracle says yes on $(x,j)$ exactly when $B_0(F_j(x))\not\equiv 0\pmod 2$. Thus the oracle answer on $(x,j)$ is precisely $s(F(x)+j)$.

A deterministic polynomial time oracle machine obtains the four bits $t(F(x))$, $s(F(x))$, $t(F(x)+1)$, and $s(F(x)+1)$. It then computes $t(F(x))\oplus s(F(x))$ and $t(F(x)+1)\oplus s(F(x)+1)$. The machine accepts if these two bits are different, and rejects otherwise. By the first theorem, the two bits are different if and only if $\ell(F(x))\not\equiv \ell(F(x)+1)\pmod 2$. By the corollary, this happens if and only if $F(x)$ is a Mersenne number. Since $x\in L$ if and only if $F(x)$ is a Mersenne number, the oracle machine decides $L$.

The computation above uses both a ${\bf B_{|1|\oplus}P}$ oracle and a ${\bf B_{|0|\oplus}P}$ oracle. Since ${\bf P}^{\bf B_{|0|\oplus}P}={\bf P}^{\bf B_{|1|\oplus}P}$, queries to either type of oracle can be simulated inside deterministic polynomial time using the other type. Therefore, the same language is decidable using only a ${\bf B_{|0|\oplus}P}$ oracle, and also using only a ${\bf B_{|1|\oplus}P}$ oracle. Thus, ${\bf MNS}\subseteq {\bf P}^{\bf B_{|0|\oplus}P}$ and ${\bf MNS}\subseteq {\bf P}^{\bf B_{|1|\oplus}P}$. Also, recall that ${\bf C_{=}P}={\bf ES}={\bf MNS}$ has already been established. Then we can conclude that ${\bf C_{=}P}={\bf ES}={\bf MNS}\subseteq {\bf P}^{\bf B_{|0|\oplus}P}={\bf P}^{\bf B_{|1|\oplus}P}$.$\qed$
\end{proof}

\subsection{Containment of ${\bf PP}$ }

It was shown in \cite{P26b} that the parity of a ${\bf \#P}$ value can be recovered using polynomial time access to either ${\bf B_{|0|\oplus}P}$ or ${\bf B_{|1|\oplus}P}$ oracles. This result was used in proving ${\bf \oplus P}\subseteq {\bf P}^{{\bf B_{|0|\oplus}P}}$ and ${\bf \oplus P}\subseteq {\bf P}^{{\bf B_{|1|\oplus}P}}$. We now recover the bits of a ${\bf \#P}$ value one bit at a time by using the same four query parity subroutine, where on each query the value is padded and shifted. This upgrades the parity test to a threshold test, and gives the containment of ${\bf PP}$ in ${\bf P}^{\bf B_{|0|\oplus}P}$ and ${\bf P}^{\bf B_{|1|\oplus}P}$. Consequently, we get ${\bf P}^{\bf PP}={\bf P}^{\bf B_{|0|\oplus}P}={\bf P}^{\bf B_{|1|\oplus}P}$.

\begin{theorem}
${\bf PP}\subseteq{\bf P}^{\bf B_{|0|\oplus}P}$
\end{theorem}

\begin{proof}
Let $L\in{\bf PP}$. Then there is a non-deterministic polynomial time machine $M$ and a polynomial $p$ such that, on every input $x$, the machine $M(x)$ uses exactly $p(|x|)$ non-deterministic bits, and $x\in L$ if and only if $M(x)$ has more accepting paths than rejecting paths. Let $A(x)=\#\operatorname{acc}_M(x)$. Since $M(x)$ has exactly $2^{p(|x|)}$ computation paths, we have $x\in L$ if and only if $A(x)>2^{p(|x|)-1}$. We may assume without loss of generality that $p(|x|)\geq 1$ by adding one non-deterministic bit if necessary.

Set $p=p(|x|)$ and $n=p+1$. Since $0\leq A(x)\leq 2^p<2^n$, there are uniquely determined bits $a_0,a_1,\ldots,a_{n-1}\in\{0,1\}$ such that $A(x)=\sum_{i=0}^{n-1}a_i2^i$. Define $X(x)=2^n+A(x)$. Then the standard binary representation of $X(x)$ has the form $(1a_{n-1}a_{n-2}\cdots a_1a_0)_2$.

We next show that a deterministic polynomial time machine with oracle access to a language in ${\bf B_{|0|\oplus}P}$ can recover the bits $a_0,a_1,\ldots,a_{n-1}$ one at a time. Assume that the first $r$ low order bits $a_0,\ldots,a_{r-1}$ have already been determined. We then write $X(x)=2^rQ_r(x)+b_r$, where $b_r=X(x)\bmod 2^r$. And the value $b_r$ is already known from the recovered low order bits. The next desired bit is $a_r=Q_r(x)\bmod 2$. For $j\in\{0,1,2,3\}$, consider the function $X(x)+j2^r$. We restrict the tagged inputs $(x,r,j)$ to $0\leq r<n$ and $j\in\{0,1,2,3\}$. On these inputs, $X(x)+j2^r$ is a ${\bf \# P}$ function, since $A(x)$ is a ${\bf \# P}$ function, the terms $2^n$ and $j2^r$ are nonnegative polynomial time computable functions with polynomially bounded binary length, and ${\bf \# P}$ is closed under addition. For invalid tags, the underlying counting function is defined to have value $1$, so that the zero bit parity oracle rejects them, since $B_0(1)=0$. As a result, we obtain a single tagged oracle language $O\in{\bf B_{|0|\oplus}P}$, where a valid query $(x,r,j)$ asks whether $B_0(X(x)+j2^r)$ is odd.

We now have $X(x)+j2^r=2^r(Q_r(x)+j)+b_r$. Since $Q_r(x)+j>0$, the binary representation of $X(x)+j2^r$ consists of the binary representation of $Q_r(x)+j$ followed by its lower $r$ bit positions, whose value is $b_r$. Let $z_r$ be the number of $0$'s bits among these lower $r$ positions. Since $b_r$ is known then $z_r$ is also known. As a result, $B_0(X(x)+j2^r)\equiv B_0(Q_r(x)+j)+z_r\pmod 2$. Thus, the deterministic machine can compute $s(Q_r(x)+j)=B_0(Q_r(x)+j)\pmod 2$ from the oracle answer for $X(x)+j2^r$.

So the four oracle queries for $j$ equal to $0,1,2,3$ give the four values $s(Q_r(x))$, $s(Q_r(x)+1)$, $s(Q_r(x)+2)$, and $s(Q_r(x)+3)$. By the four consecutive value theorem for $s$, these four values determine whether $Q_r(x)$ is even or odd. \cite{P26b} Thus, they determine $a_r=Q_r(x)\bmod 2$. Repeating this procedure for $r=0,1,\ldots,n-1$ recovers all $n$ bits of $A(x)$. This uses $4n$ oracle query calls, and since $n=p(|x|)+1$ is polynomially bounded, this is polynomially many queries. Once $A(x)$ is recovered, the deterministic machine checks whether $A(x)>2^{p(|x|)-1}$ and accepts exactly in that case.

Therefore, $L$ is decidable by a deterministic polynomial time machine with oracle access to a language in ${\bf B_{|0|\oplus}P}$. Since $L\in{\bf PP}$ was for any $L$, we can then conclude that ${\bf PP}\subseteq{\bf P}^{\bf B_{|0|\oplus}P}$.$\qed$
\end{proof}

\begin{theorem}
${\bf PP}\subseteq{\bf P}^{\bf B_{|1|\oplus}P}$
\end{theorem}

\begin{proof}
Let $L\in{\bf PP}$. Then there is a non-deterministic polynomial time machine $M$ and a polynomial $p$ such that, on every input $x$, the machine $M(x)$ uses exactly $p(|x|)$ non-deterministic bits, and $x\in L$ if and only if $M(x)$ has more accepting paths than rejecting paths. Let $A(x)=\#\operatorname{acc}_M(x)$. Since $M(x)$ has exactly $2^{p(|x|)}$ computation paths, we have $x\in L$ if and only if $A(x)>2^{p(|x|)-1}$. We assume without loss of generality that $p(|x|)\geq 1$ by adding one non-deterministic bit if necessary. 

Set $p=p(|x|)$ and $n=p+1$. Since $0\leq A(x)\leq 2^p<2^n$, there are uniquely determined bits $a_0,a_1,\ldots,a_{n-1}\in\{0,1\}$ such that $A(x)=\sum_{i=0}^{n-1}a_i2^i$. Define $X(x)=2^n+A(x)$. Then the standard binary representation of $X(x)$ has the form $(1a_{n-1}a_{n-2}\cdots a_1a_0)_2$.

We next show that a deterministic polynomial time machine with oracle access to a language in ${\bf B_{|1|\oplus}P}$ can recover the bits $a_0,a_1,\ldots,a_{n-1}$ one at a time. Assume that the first $r$ low order bits $a_0,\ldots,a_{r-1}$ have already been determined. We then write $X(x)=2^rQ_r(x)+b_r$, where $b_r=X(x)\bmod 2^r$. And the value $b_r$ is already known from the recovered low order bits. The next desired bit is $a_r=Q_r(x)\bmod 2$. For $j\in\{0,1,2,3\}$, consider the function $X(x)+j2^r$. We restrict the tagged inputs $(x,r,j)$ to $0\leq r<n$ and $j\in\{0,1,2,3\}$. On these inputs, $X(x)+j2^r$ is a ${\bf \#P}$ function, since $A(x)$ is a ${\bf \#P}$ function, the terms $2^n$ and $j2^r$ are nonnegative polynomial time computable functions with polynomially bounded binary length, and ${\bf \#P}$ is closed under addition. For invalid tags, the underlying counting function is defined to have value $0$, so that the one bit parity oracle rejects them, since $B_1(0)=0$. As a result, we obtain a single tagged oracle language $O\in{\bf B_{|1|\oplus}P}$, where a valid query $(x,r,j)$ asks whether $B_1(X(x)+j2^r)$ is odd.

We now have $X(x)+j2^r=2^r(Q_r(x)+j)+b_r$. Since $Q_r(x)+j>0$, the binary representation of $X(x)+j2^r$ consists of the binary representation of $Q_r(x)+j$ followed by its lower r bit positions whose value is $b_r$. Let $u_r$ be the number of $1$'s bits among these lower $r$ positions. Since $b_r$ is known then $u_r$ is also known. As a result, $B_1(X(x)+j2^r)\equiv B_1(Q_r(x)+j)+u_r\pmod 2$. Thus, the deterministic machine can compute $t(Q_r(x)+j)=B_1(Q_r(x)+j)\pmod 2$ from the oracle answer for $X(x)+j2^r$. 

So the four oracle queries for $j$ equal to $0,1,2,3$ give the four values $t(Q_r(x))$, $t(Q_r(x)+1)$, $t(Q_r(x)+2)$, and $t(Q_r(x)+3)$. By the four consecutive value theorem for $t$, these four values determine whether $Q_r(x)$ is even or odd.\cite{P26b} Thus they determine $a_r=Q_r(x)\bmod 2$. Repeating this procedure for $r=0,1,\ldots,n-1$ recovers all $n$ bits of $A(x)$. This uses $4n$ oracle query calls, and since $n=p(|x|)+1$ is polynomially bounded, this is polynomially many queries. Once $A(x)$ is recovered, the deterministic machine checks whether $A(x)>2^{p(|x|)-1}$ and  accepts exactly in that case.

Therefore $L$ is decidable by a deterministic polynomial time machine with oracle access to a language in ${\bf B_{|1|\oplus}P}$. Since $L\in{\bf PP}$ was for any $L$, we can then conclude that ${\bf PP}\subseteq{\bf P}^{\bf B_{|1|\oplus}P}$. $\qed$
\end{proof}

\begin{theorem}
${\bf P}^{\bf PP}={\bf P}^{\bf B_{|0|\oplus}P}={\bf P}^{\bf B_{|1|\oplus}P}$
\end{theorem}

\begin{proof}
We already proved that ${\bf PP}\subseteq{\bf P}^{\bf B_{|0|\oplus}P}$. Therefore, a deterministic polynomial time machine with oracle access to a ${\bf PP}$ language can simulate each ${\bf PP}$ oracle query by a deterministic polynomial time computation with oracle access to a language in ${\bf B_{|0|\oplus}P}$. We then obtain ${\bf P}^{\bf PP}\subseteq{\bf P}^{{\bf P}^{\bf B_{|0|\oplus}P}}={\bf P}^{\bf B_{|0|\oplus}P}$ since nested deterministic polynomial time oracle computations collapse. 

For the reverse containment, we already know that ${\bf B_{|0|\oplus}P}\subseteq{\bf P}^{\bf PP}$. Therefore, a deterministic polynomial time machine with oracle access to a language in ${\bf B_{|0|\oplus}P}$ can simulate each such oracle query by a deterministic polynomial time computation with oracle access to ${\bf PP}$. We then obtain ${\bf P}^{\bf B_{|0|\oplus}P}\subseteq{\bf P}^{{\bf P}^{\bf PP}}={\bf P}^{\bf PP}$ since nested deterministic polynomial time oracle computations collapse. 

As a result, we obtain ${\bf P}^{\bf PP}={\bf P}^{\bf B_{|0|\oplus}P}$. Proof of ${\bf P}^{\bf PP}={\bf P}^{\bf B_{|1|\oplus}P}$ is identical, but we already know that ${\bf P}^{\bf B_{|0|\oplus}P}={\bf P}^{\bf B_{|1|\oplus}P}$.\cite{P26a} Therefore, we can conclude that ${\bf P}^{\bf PP}={\bf P}^{\bf B_{|0|\oplus}P}={\bf P}^{\bf B_{|1|\oplus}P}$. $\qed$
\end{proof} 

When we combine this result with the result from \cite{P26a} we obtain the following equality: ${\bf P}^{\bf PP}={\bf P}^{{\bf B_{|0|=|1|}P}}={\bf P}^{{\bf B_{|0|>|1|}P}}={\bf P}^{{\bf B_{|0|<|1|}P}}= {\bf P}^{{\bf B_{|0|\oplus}P}} = {\bf P}^{{\bf B_{|1|\oplus}P}}$. In other words, both comparison based and parity based bit-counting complexity classes are equivalent to ${\bf PP}$ under polynomial time Turing reductions.

\subsection{Containment of \# P}
The first two theorems in the previous subsection actually recovered all of the bits of a ${\bf \#P}$ value. Therefore, we can indeed obtain the ${\bf \#P}$ value if we change the base machine from ${\bf P}$ to ${\bf FP}$. This is what we accomplish with the following two theorems and prove that ${\bf \#P}\subseteq{\bf FP}^{\bf B_{|0|\oplus}P}$ and ${\bf \#P}\subseteq{\bf FP}^{\bf B_{|1|\oplus}P}$. 

\begin{theorem}
${\bf \#P}\subseteq{\bf FP}^{\bf B_{|0|\oplus}P}$
\end{theorem}

\begin{proof}
Let $F\in{\bf \#P}$. Then there is a polynomial time predicate $R(x,y)$ and a polynomial $p$ such that, for every input $x$, $F(x)=||{y\in\{0,1\}^{p(|x|)}\mid R(x,y)}||$. Set $p=p(|x|)$ and $n=p+1$. Then $0\leq F(x)\leq 2^p<2^n$. As a result, there are uniquely determined bits $a_0,a_1,\ldots,a_{n-1}\in\{0,1\}$ such that $F(x)=\sum_{i=0}^{n-1}a_i2^i$. We next define $X(x)=2^n+F(x)$, where the standard binary representation of $X(x)$ has the form
$(1a_{n-1}a_{n-2}\cdots a_1a_0)_2$.

Let $s(N)=B_0(N)\pmod 2$. We use the four consecutive value theorem for $s$, where for every $N\geq 0$, the parity of $N$ is determined by the four values $s(N)$, $s(N+1)$, $s(N+2)$, and $s(N+3)$. Specifically, $N$ is even if and only if $s(N)\neq s(N+1)$ and $s(N+2)\neq s(N+3)$.\cite{P26b}

We recover the bits $a_0,a_1,\ldots,a_{n-1}$ one at a time. Assume that the first $r$ low order bits $a_0,\ldots,a_{r-1}$ have already been recovered. We next write $X(x)=2^rQ_r(x)+b_r$, where $b_r=X(x)\bmod 2^r$. Since the first $r$ low order bits of $X(x)$ are already known, $b_r$ is known. The next desired bit is $a_r=Q_r(x)\bmod 2$. For each $j\in\{0,1,2,3\}$, consider the function $X(x)+j2^r$. We restrict the tagged inputs $(x,r,j)$ to $0\leq r<n$ and $j\in\{0,1,2,3\}$. On these inputs, $X(x)+j2^r$ is a  ${\bf \#P}$ function, since $F(x)$ is a  ${\bf \#P}$ function and the terms $2^n$ and $j2^r$ are nonnegative polynomial time computable functions with polynomially bounded binary length. For invalid tags, the underlying counting function is defined to have value $1$, so that the zero bit parity oracle rejects them. As a result, the tagged language $O_s=\{(x,r,j)\mid B_0(X(x)+j2^r)\not\equiv 0\pmod 2\}$ belongs to ${\bf B_{|0|\oplus}P}$.

We now have $X(x)+j2^r=2^r(Q_r(x)+j)+b_r$. Since $Q_r(x)+j>0$, the standard binary representation of $X(x)+j2^r$ consists of the binary representation of $Q_r(x)+j$ followed by the lower $r$ bit positions whose numerical value is $b_r$. Let $z_r$ be the number of $0$'s bits among these lower $r$ bit positions. Since $b_r$ is known, $z_r$ is also known. As a result, $B_0(X(x)+j2^r)\equiv B_0(Q_r(x)+j)+z_r\pmod 2$. Thus the oracle answer for $(x,r,j)$, corrected by the known bit $z_r\pmod 2$, gives $s(Q_r(x)+j)$. Then making the four oracle queries for $j$ equal to $0,1,2,3$ gives $s(Q_r(x))$, $s(Q_r(x)+1)$, $s(Q_r(x)+2)$, and $s(Q_r(x)+3)$. These four values determine whether $Q_r(x)$ is even or odd by the four consecutive value theorem for $s$. Thus they determine $a_r=Q_r(x)\bmod 2$.

Repeating this procedure for $r=0,1,\ldots,n-1$ recovers all $n$ bits of $F(x)$. Since $n=p(|x|)+1$ is polynomially bounded, this is a deterministic polynomial time oracle computation. After the bits $a_0,a_1,\ldots,a_{n-1}$ are recovered, the machine computes the integer $F(x)=\sum_{i=0}^{n-1}a_i2^i$ and outputs its standard binary encoding. If all recovered bits are $0$, then the machine outputs the standard encoding of $0$. Therefore the value of the ${\bf \#P}$ function $F$ is computed by a deterministic polynomial time oracle transducer with access to oracle ${\bf B_{|0|\oplus}P}$. Since $F\in{\bf \#P}$ was for any $F$, we can then conclude that ${\bf \#P}\subseteq{\bf FP}^{\bf B_{|0|\oplus}P}$.$\qed$
\end{proof}

\begin{theorem}
${\bf \#P}\subseteq{\bf FP}^{\bf B_{|1|\oplus}P}$
\end{theorem}

\begin{proof}
Let $F\in{\bf \#P}$. Then there is a polynomial time predicate $R(x,y)$ and a polynomial $p$ such that, for every input $x$, $F(x)=||{y\in\{0,1\}^{p(|x|)}\mid R(x,y)}||$. Set $p=p(|x|)$ and $n=p+1$. Then $0\leq F(x)\leq 2^p<2^n$. As a result, there are uniquely determined bits $a_0,a_1,\ldots,a_{n-1}\in\{0,1\}$ such that $F(x)=\sum_{i=0}^{n-1}a_i2^i$. We next define $X(x)=2^n+F(x)$, where the standard binary representation of $X(x)$ has the form
$(1a_{n-1}a_{n-2}\cdots a_1a_0)_2$.

Let $t(N)=B_1(N)\pmod 2$. We use the four consecutive value theorem for $t$ where for every $N\geq 0$, the parity of $N$ is determined by the four values $t(N)$, $t(N+1)$, $t(N+2)$, and $t(N+3)$. Specifically, $N$ is even if and only if $t(N)\neq t(N+1)$ and $t(N+2)\neq t(N+3)$.\cite{P26b}

We recover the bits $a_0,a_1,\ldots,a_{n-1}$ one at a time. Assume that the first $r$ low order bits $a_0,\ldots,a_{r-1}$ have already been recovered. We next write $X(x)=2^rQ_r(x)+b_r$, where $b_r=X(x)\bmod 2^r$. Since the first $r$ low order bits of $X(x)$ are already known, $b_r$ is known. The next desired bit is $a_r=Q_r(x)\bmod 2$. For each $j\in\{0,1,2,3\}$, consider the function $X(x)+j2^r$. We restrict the tagged inputs $(x,r,j)$ to $0\leq r<n$ and $j\in\{0,1,2,3\}$. On these inputs, $X(x)+j2^r$ is a ${\bf \#P}$ function, since $F(x)$ is a ${\bf \#P}$ function and the terms $2^n$ and $j2^r$ are nonnegative polynomial time computable functions with polynomially bounded binary length. For invalid tags, the underlying counting function is defined to have value $0$, so that the one bit parity oracle rejects them. As a result, the tagged language $O_t=\{(x,r,j)\mid B_1(X(x)+j2^r)\not\equiv 0\pmod 2\}$ belongs to ${\bf B_{|1|\oplus}P}$.

We now have $X(x)+j2^r=2^r(Q_r(x)+j)+b_r$. Since $Q_r(x)+j>0$, the standard binary representation of $X(x)+j2^r$ consists of the binary representation of $Q_r(x)+j$ followed by the lower $r$ bit positions whose numerical value is $b_r$. Let $u_r$ be the number of $1$'s bits among these lower $r$ bit positions. Since $b_r$ is known, $u_r$ is also known. As a result, $B_1(X(x)+j2^r)\equiv B_1(Q_r(x)+j)+u_r\pmod 2$. Thus the oracle answer for $(x,r,j)$, corrected by the known bit $u_r\pmod 2$, gives $t(Q_r(x)+j)$. Then making the four oracle queries for $j$ equal to $0,1,2,3$ gives $t(Q_r(x))$, $t(Q_r(x)+1)$, $t(Q_r(x)+2)$, and $t(Q_r(x)+3)$. These four values determine whether $Q_r(x)$ is even or odd by the four consecutive value theorem for $t$. Thus they determine $a_r=Q_r(x)\bmod 2$.

Repeating this procedure for $r=0,1,\ldots,n-1$ recovers all $n$ bits of $F(x)$. Since $n=p(|x|)+1$ is polynomially bounded, this is a deterministic polynomial time oracle computation. After the bits $a_0,a_1,\ldots,a_{n-1}$ are recovered, the machine computes the integer $F(x)=\sum_{i=0}^{n-1}a_i2^i$ and outputs its standard binary encoding. If all recovered bits are $0$, then the machine outputs the standard encoding of $0$. Therefore the value of the ${\bf \#P}$ function $F$ is computed by a deterministic polynomial time oracle transducer with access to oracle ${\bf B_{|1|\oplus}P}$. Since $F\in{\bf \#P}$ was for any $F$, we can then conclude that ${\bf \#P}\subseteq{\bf FP}^{\bf B_{|1|\oplus}P}$.$\qed$
\end{proof}

\subsection{Distinct oracle queries}

The theorems that showed ${\bf PP}\subseteq{\bf P}^{\bf B_{|0|\oplus}P}$, ${\bf PP}\subseteq{\bf P}^{\bf B_{|1|\oplus}P}$, ${\bf \#P}\subseteq{\bf FP}^{\bf B_{|0|\oplus}P}$ and ${\bf \#P}\subseteq{\bf FP}^{\bf B_{|1|\oplus}P}$ in the previous two subsections relied upon the four consecutive value theorems that were proven and used in \cite{P26b} to recover the parity bit. However, there are repeated oracle queries when you directly employ this method to recover all of the bits of a ${\bf \#P}$ function. We next prove that $2n+2$ distinct oracle queries suffice if you employ caching. 

\begin{theorem}
The $4n$ oracle queries can be replaced by $2n+2$ distinct oracle queries in the bit recovery procedure. 
\end{theorem}

\begin{proof}
In round $r$, where $0\leq r\leq n-1$, the four queried values are
$X$, $X+2^r$, $X+2^{r+1}$, and $X+3\cdot 2^r$. Thus the offsets queried in round $r$ are $0$, $2^r$, $2^{r+1}$, and $3\cdot 2^r$. To count distinct oracle queries, we use an equivalent tagged oracle whose input is $(x,c)$, where $c$ is the offset added to $X(x)$. For either $i\in\{0,1\}$, the corresponding oracle answers whether $B_i(X(x)+c)$ is odd. We restrict valid offsets to $0\leq c\leq 3\cdot 2^{n-1}$. On these inputs, $X(x)+c$ is a ${\bf \#P}$ function of the tagged input. On invalid tags, the underlying count is defined to be $1$ when $i=0$ and $0$ when $i=1$, ensuring that such tags are rejected. Each original query $(x,r,j)$ can therefore be replaced by the query $(x,j2^r)$, while the correction for the known lower order bits is still performed by the deterministic machine.

Across all rounds $r=0,1,\ldots,n-1$, the offset $0$ appears in every round, but needs to be queried only once. The offsets of the form $2^r$ or $2^{r+1}$ together give exactly $2^0,2^1,\ldots,2^n$, which are $n+1$ distinct offsets. The offsets of the form $3\cdot 2^r$ give exactly $3\cdot 2^0,3\cdot 2^1,\ldots,3\cdot 2^{n-1}$, which are $n$ distinct offsets. Also, no offset of the form $3\cdot 2^r$ is equal to a power of $2$, because $3\cdot 2^r$ has odd factor $3$. So the total number of distinct offsets is $1+(n+1)+n=2n+2$. Thus, the required offsets form the set $\{0\}\cup\{2^0,2^1,\ldots,2^n\}\cup\{3\cdot2^0,3\cdot2^1,\ldots,3\cdot2^{n-1}\}$. The deterministic machine queries the offset based oracle once for each of these $2n+2$ offsets, caches the answers, and reuses them in the appropriate rounds. Therefore, the bit recovery procedure can be performed using only $2n+2$ distinct oracle queries. $\qed$
\end{proof}

Note that the above theorem shows a cached query upper bound for the bit recovery procedure. It does not claim that $2n+2$ is the minimum possible number of distinct oracle queries.

\section{Properties of parity based bit-counting hierarchies}

We established in the previous section that ${\bf PP}\subseteq{\bf P}^{\bf B_{|0|\oplus}P}$, ${\bf PP}\subseteq{\bf P}^{\bf B_{|1|\oplus}P}$ and ${\bf P}^{\bf PP}={\bf P}^{\bf B_{|0|\oplus}P}={\bf P}^{\bf B_{|1|\oplus}P}$. We next prove the following two theorems that will guide us in showing that our parity based bit-counting hierarchies contain counting hierarchy. 

\begin{theorem}
${\bf PP}^{\bf A}\subseteq{\bf P}^{({\bf B_{|i|\oplus}P})^{\bf A}}$  for every oracle class ${\bf A}$, where $i\in \{0,1\}$
\end{theorem}

\begin{proof}
The proofs of  ${\bf PP}\subseteq{\bf P}^{\bf B_{|0|\oplus}P}$ and ${\bf PP}\subseteq{\bf P}^{\bf B_{|1|\oplus}P}$ work by recovering the accepting path count with polynomially many bit-counting parity queries. These proofs relativize because if the original ${\bf PP}$ machine has oracle access to a language in ${\bf A}$, then its accepting path count is a ${\bf \#P}^{\bf A}$ count, and the same bit recovery construction uses queries to $({\bf B_{|0|\oplus}P})^{\bf A}$ or $({\bf B_{|1|\oplus}P})^{\bf A}$. Therefore, ${\bf PP}^{\bf A}\subseteq{\bf P}^{({\bf B_{|0|\oplus}P})^{\bf A}}$ and ${\bf PP}^{\bf A}\subseteq{\bf P}^{({\bf B_{|1|\oplus}P})^{\bf A}}$.$\qed$
\end{proof}

\begin{theorem}
${\bf P}^{\bf A}\subseteq({\bf B_{|i|\oplus}P})^{\bf A}$ for every oracle class ${\bf A}$, where $i\in \{0,1\}$
\end{theorem}

\begin{proof}
Let $L\in{\bf P}^{\bf A}$. Since deterministic polynomial time oracle classes are closed under complement, $\overline{L}\in{\bf P}^{\bf A}$ as well.

To show $L\in({\bf B_{|0|\oplus}P})^{\bf A}$, define a ${\bf \#P}^{\bf A}$ function $H_0(x)$ by letting the machine have one computation path and accepting on that path exactly when $x\notin L$. Thus $H_0(x)=0$ if $x\in L$, and $H_0(x)=1$ if $x\notin L$. Since $B_0(0)=1$ is odd and $B_0(1)=0$ is even, we have $x\in L$ if and only if $B_0(H_0(x))\not\equiv 0\pmod 2$. Therefore, $L\in({\bf B_{|0|\oplus}P})^{\bf A}$.

To show $L\in({\bf B_{|1|\oplus}P})^{\bf A}$, define a ${\bf \#P}^{\bf A}$ function $H_1(x)$ by letting the machine have one computation path and accepting on that path exactly when $x\in L$. Thus $H_1(x)=1$ if $x\in L$, and $H_1(x)=0$ if $x\notin L$. Since $B_1(1)=1$ is odd and $B_1(0)=0$ is even, we have $x\in L$ if and only if $B_1(H_1(x))\not\equiv 0\pmod 2$. Therefore, $L\in({\bf B_{|1|\oplus}P})^{\bf A}$.$\qed$
\end{proof}

\begin{theorem}
${\bf CH}\subseteq{\bf \Delta B_{|0|\oplus}H}={\bf \Delta B_{|1|\oplus}H}$.
\end{theorem}

\begin{proof}
We prove the level by level containment ${\bf C_kP}\subseteq{\bf \Delta B_{|0|\oplus}^{k}}$ for every $k\geq 0$, where ${\bf C_kP}$ denotes the $k$th level of counting hierarchy. For $k=0$, we have ${\bf C_0P}={\bf P}={\bf \Delta B_{|0|\oplus}^{0}}$, so the containment is immediate.

We next  assume ${\bf C_kP}\subseteq{\bf \Delta B_{|0|\oplus}^{k}}$. By definition of counting hierarchy, ${\bf C_{k+1}P}={\bf PP}^{\bf C_kP}$. By the relativized simulation theorem 9, ${\bf PP}^{\bf C_kP}\subseteq{\bf P}^{({\bf B_{|0|\oplus}P})^{\bf C_kP}}$. Using the induction hypothesis, oracle access to ${\bf C_kP}$ can be simulated by oracle access to ${\bf \Delta B_{|0|\oplus}^{k}}$. Thus ${\bf C_{k+1}P}\subseteq{\bf P}^{({\bf B_{|0|\oplus}P})^{{\bf \Delta B_{|0|\oplus}^{k}}}}$.

By definition, ${\bf \Delta B_{|0|\oplus}^{k}}={\bf P}^{{\bf \Sigma B_{|0|\oplus}^{k}}}$. An oracle computation in ${\bf P}^{{\bf \Sigma B_{|0|\oplus}^{k}}}$ is performed by a deterministic polynomial time machine with oracle access to a language in ${{\bf \Sigma B_{|0|\oplus}^{k}}}$. Each such computation can be simulated directly by the underlying non-deterministic bit-counting machine using the same oracle. Since this simulation is deterministic, it does not introduce additional non-deterministic computation paths and therefore preserves the accepting path count. Consequently, deterministic polynomial time oracle computations can be absorbed into the oracle access, and we have $({\bf B_{|0|\oplus}P})^{{\bf \Delta B_{|0|\oplus}^{k}}} =({\bf B_{|0|\oplus}P})^{{\bf P}^{{\bf \Sigma B_{|0|\oplus}^{k}}}} \subseteq({\bf B_{|0|\oplus}P})^{{\bf \Sigma B_{|0|\oplus}^{k}}} ={\bf \Sigma B_{|0|\oplus}^{k+1}}$. Therefore, ${\bf C_{k+1}P}\subseteq{\bf P}^{{\bf \Sigma B_{|0|\oplus}^{k+1}}}
={\bf \Delta B_{|0|\oplus}^{k+1}}$. This completes the induction.

Taking unions over all $k$ gives ${\bf CH}\subseteq{\bf \Delta B_{|0|\oplus}H}$. Since we already established ${\bf \Delta B_{|0|\oplus}H}={\bf \Delta B_{|1|\oplus}H}$, we can then conclude that
${\bf CH}\subseteq{\bf \Delta B_{|0|\oplus}H}={\bf \Delta B_{|1|\oplus}H}$.$\qed$
\end{proof}

\begin{theorem}
${\bf CH}\subseteq{\bf \Sigma B_{|0|\oplus}H}={\bf \Sigma B_{|1|\oplus}H}$.
\end{theorem}

\begin{proof}
From the preceding theorem, for every $k\geq 0$ we have ${\bf C_kP}\subseteq{\bf \Delta B_{|0|\oplus}^{k}}$. Also, for every $k\geq 0$, ${\bf \Delta B_{|0|\oplus}^{k}}={\bf P}^{{\bf \Sigma B_{|0|\oplus}^{k}}}$. By the relativized simulation theorem 10, ${\bf P}^{{\bf \Sigma B_{|0|\oplus}^{k}}}\subseteq({\bf B_{|0|\oplus}P})^{{\bf \Sigma B_{|0|\oplus}^{k}}}$. By definition, $({\bf B_{|0|\oplus}P})^{{\bf \Sigma B_{|0|\oplus}^{k}}} ={\bf \Sigma B_{|0|\oplus}^{k+1}}$. Therefore, ${\bf \Delta B_{|0|\oplus}^{k}}\subseteq{\bf \Sigma B_{|0|\oplus}^{k+1}}$.

Combining the two containments gives ${\bf C_kP}\subseteq{\bf \Sigma B_{|0|\oplus}^{k+1}}$ for every $k\geq 0$. Taking unions over all $k$ gives ${\bf CH}\subseteq{\bf \Sigma B_{|0|\oplus}H}$. Since we already established ${\bf \Sigma B_{|0|\oplus}H}={\bf \Sigma B_{|1|\oplus}H}$, we can then conclude that ${\bf CH}\subseteq{\bf \Sigma B_{|0|\oplus}H}={\bf \Sigma B_{|1|\oplus}H}$.$\qed$
\end{proof}

Combining these results with the ones in \cite{P26b} establishes that counting hierarchy and parity based bit-counting hierarchies coincide when their levels are united.

\section{Conclusion}

We studied several additional properties of parity based bit-counting complexity classes ${\bf B_{|0|\oplus}P}$ and ${\bf B_{|1|\oplus}P}$. We first established a relationship between the parity of the binary length of an integer and the parity of its number of 0's bits and 1's bits. Using this fact, we showed that the change in the binary length parity between an integer and its successor characterizes Mersenne numbers. This allowed us to prove that ${\bf C_{=}P}={\bf ES}={\bf MNS}\subseteq{\bf P}^{\bf B_{|0|\oplus}P}={\bf P}^{\bf B_{|1|\oplus}P}$. We did this by noticing that the bit length of an integer changes exactly when you add a $1$ to a Mersenne number and that this bit length difference is detectable by querying ${\bf B_{|0|\oplus}P}$ and ${\bf B_{|1|\oplus}P}$ oracles.

We then extended the four consecutive value theorems that were proven in \cite{P26b} to show that ${\bf \oplus P}\subseteq {\bf P}^{{\bf B_{|0|\oplus}P}}$ and ${\bf \oplus P}\subseteq {\bf P}^{{\bf B_{|1|\oplus}P}}$. We recovered the ${\bf \#P}$ value one bit at a time by padding the counting value, querying appropriately shifted versions of that value, and removing the known contributions of previously recovered bits. This enabled us to not only prove that ${\bf PP}\subseteq {\bf P}^{{\bf B_{|0|\oplus}P}}$ and ${\bf PP}\subseteq {\bf P}^{{\bf B_{|1|\oplus}P}}$, but also prove that ${\bf \#P}\subseteq {\bf FP}^{{\bf B_{|0|\oplus}P}}$ and ${\bf \#P}\subseteq {\bf FP}^{{\bf B_{|1|\oplus}P}}$. We also demonstrated that the straightforward $4n$ query procedure can be implemented using only $2n+2$ distinct queried numerical values by caching repeated answers. The importance of this construction is that parity based bit-counting oracles can be used not merely to recover the parity of a counting function, but to reconstruct its complete numerical output.

We then showed that ${\bf P}^{\bf PP}= {\bf P}^{{\bf B_{|0|\oplus}P}} = {\bf P}^{{\bf B_{|1|\oplus}P}}$. When this result is combined with the result from \cite{P26a}, we obtained the following equivalence:  ${\bf P}^{\bf PP}={\bf P}^{{\bf B_{|0|=|1|}P}}={\bf P}^{{\bf B_{|0|>|1|}P}}={\bf P}^{{\bf B_{|0|<|1|}P}}= {\bf P}^{{\bf B_{|0|\oplus}P}} = {\bf P}^{{\bf B_{|1|\oplus}P}}$. Essentially, parity based bit-counting complexity classes and comparison based bit-counting complexity classes are equivalent under polynomial time Turing reductions not only among themselves, but also to ${\bf PP}$. These results are particularly interesting because the underlying acceptance criteria examine substantially different properties of a counting function. After all, there is quite a bit of difference between determining whether a Boolean formula in CNF has 1) more satisfying truth assignments than falsifying truth assignments, or 2) the number of satisfying truth assignments in its binary expansion has i) equal number of 0's bits and 1's bits, or ii) more number of 0's bits than 1's bits, or iii) less number of 0's bits than 1's bits, or iv) odd number of 0's bits, or v) odd number of 1's bits. At first glance, it seemed improbable for these languages to have any equivalence under any scenario, but then in an absolutely amazing way they ended up being equivalent under polynomial time Turing reductions. This further illustrates the observation in [KPZ99] that polynomial time Turing reductions can blur structural differences between complexity classes.

We also investigated parity based bit-counting hierarchies and proved that counting hierarchy is contained in parity based bit-counting hierarchies ${\bf \Sigma B_{|0|\oplus}H}$, ${\bf \Sigma B_{|1|\oplus}H}$, ${\bf \Delta B_{|0|\oplus}H}$ and ${\bf \Delta B_{|1|\oplus}H}$. Combining these containment results with the reverse containment results established in \cite{P26b}, we can conclude that counting hierarchy and parity based bit-counting hierarchies coincide when their levels are united. 

It is also worth contrasting the bit-counting complexity classes with the complexity class ${\bf MP}$ introduced in \cite{GKRST95}. The complete problem for ${\bf MP}$ requires two distinct pieces of numerical information: the number of accepting computation paths of a non-deterministic polynomial time Turing machine and a separately specified bit position in the binary representation of that number. In contrast, all of the bit-counting complexity classes have acceptance criteria that depend exclusively on the numerical value of a single ${\bf \# P}$ function. Consequently, each such acceptance criterion determines a fixed set of nonnegative integers, which can be represented by an integer sequence enumerating precisely those values that satisfy the criterion. In essence, for the bit-counting complexity classes ${\bf B_{|0|\oplus}P}$, ${\bf B_{|1|\oplus}P}$, ${\bf B_{|0|=|1|}P}$, ${\bf B_{|0|>|1|}P}$, and ${\bf B_{|0|<|1|}P}$, the corresponding integer sequences that define their acceptance criteria are \cite{P00}, \cite{SP96}, \cite{K99}, \cite{Z02a}, and \cite{Z02b}, respectively. This distinction is better illustrated with the following example. Assume that the output value of a ${\bf \#P}$ function is $(101011)_2=43$, which contains two 0's bits and four 1's bits. Therefore, 43 does not satisfy the acceptance criteria of ${\bf B_{|0|\oplus}P}$,  ${\bf B_{|1|\oplus}P}$, ${\bf B_{|0|=|1|}P}$ or  ${\bf B_{|0|>|1|}P}$, but it does satisfy the acceptance criterion of ${\bf B_{|0|<|1|}P}$. However, an instance of the MP-complete bit selection problem with the same counting value may be accepted or rejected depending on the separately specified bit position. Thus, unlike the bit-counting complexity classes considered here, the acceptance criterion of the MP-complete problem cannot be characterized by any single integer sequence enumerating accepted counting values alone. This constitutes a fundamental structural distinction between the acceptance criteria of ${\bf MP}$ and those of the bit-counting complexity classes, despite their common reliance on the binary representations of ${\bf \# P}$ outputs.

As a next step, we can study the semantic variants of these bit-counting complexity classes as it was previously noted in \cite{P26a}. Or we can investigate the logspace variants of these bit-counting complexity classes as it was previously noted in \cite{P26b}. The bit recovery construction also raises questions about oracle query efficiency, particularly whether the $2n+2$ distinct query upper bound can be improved through a different recovery method. Although the five bit-counting complexity classes ${\bf B_{|0|\oplus}P}$, ${\bf B_{|1|\oplus}P}$, ${\bf B_{|0|=|1|}P}$, ${\bf B_{|0|>|1|}P}$, and ${\bf B_{|0|<|1|}P}$ are all equivalent under polynomial time Turing reductions, their direct containment relationships are not fully understood. We only know that ${\bf B_{|1|\oplus}P}\subseteq{\bf B_{|0|\oplus}P}$, which was proven in \cite{P26b}. Determining whether additional containments or equalities hold among them remains an interesting direction for future research. Furthermore, perhaps there is a way to show that these bit-counting complexity classes contain the ${\bf PH}$ independent of Toda's theorem \cite{T89} as well as any other method used in it, such as the Valiant-Vazirani isolation theorem \cite{VV85}.

\bibliographystyle{alpha}
\bibliography{Even_Morel_Properties_Of_Parity_Based_Bit_Counting_Complexity_Classes2}

\end{document}